\documentclass[12pt]{amsart}
\usepackage{latexsym,amsmath,amsfonts,amscd,amssymb}
\usepackage{graphicx}
\usepackage{enumerate}
\newtheorem{theorem}{Theorem}[section]
\newtheorem{theorem*}{Theorem}

\newtheorem{proposition*}[theorem*]{Proposition}
\newtheorem{corollary}[theorem]{Corollary}
\newtheorem{corollary*}[theorem*]{Corollary}

\newtheorem{example}[theorem]{Example}

\theoremstyle{remark}

\newtheorem{remark*}[theorem*]{Remark}

\newtheorem{note*}[theorem*]{Note}

\title{Bitcoin Selfish Mining and Dyck Words}


\author[C. Grunspan]{Cyril Grunspan}
\address{Cyril Grunspan\newline{}\indent L\'eonard de Vinci P\^ole Univ, Finance Lab\newline{}\indent Paris, France, }
\email{cyril.grunspan@devinci.fr}
\author[R. P\'{e}rez-Marco]{Ricardo P\'{e}rez-Marco}
\address{Ricardo P\'{e}rez-Marco\newline{}\indent CNRS, IMJ-PRG, Univ. Paris-Diderot \newline{}\indent Paris, France}
\email{ricardo.perez.marco@gmail.com}

\begin{document}

\begin{abstract}
  We give a straightforward proof for the formula giving the long-term
  apparent hashrate of the Selfish Mining strategy in Bitcoin using only elementary
  probabilities and combinatorics, and more precisely, Dyck words. There is no need
  to compute stationary probabilities on Markov chain nor stopping times for
  Poisson processes as it was previously done. We consider also several other block
  withholding strategies.
\end{abstract}

\date{February 4, 2019}

\maketitle

\date{February 5, 2019}

\section{Introduction}

Selfish mining (in short SM) is a particular non-stop strategy of block
withholding strategy described in {\cite{ES14}} which exploits a flaw in the
Bitcoin protocol in the difficulty adjustment formula {\cite{GPM18}}. The
strategy is made of attack cycles. During each attack cycle, the attacker adds
blocks to a secret fork and then, broadcasts them to peers with an appropriate
timing. This is a deviant strategy from the Bitcoin protocol since an honest
miner never withholds blocks and always mines on top of the last block of the
official blockchain {\cite{N08}}.

As explained in {\cite{GPM18}} the good objective function based on sound
economics principles in order to compare profitabilities of mining strategies
with repetition is the revenue ratio $\frac{\mathbb{E} [R]}{\mathbb{E} [T]}$
where $R$ is the revenue of the miner per attack cyle and $T$ is the duration
time per cycle. After a difficulty adjustment, this mean duration time becomes
equal to $\mathbb{E} [L] \cdot \tau_B$ where $L$ is the number of blocks
added to the official blockchain by the network per attack cycle and $\tau_B =
600$ sec. {\cite{GPM18e}}. Thus, the objective function becomes proportional
to the long-term apparent hashrate of the strategy $\tilde{q} =
\frac{\mathbb{E} [Z]}{\mathbb{E} [L]}$ where $Z$ is the number of blocks
added by the attacker to the official blockchain per attack cycle. Several
methods have been conceived to compute \ $\tilde{q}$. In {\cite{ES14}}, first
a stationary probability is computed for a Markov chain. In {\cite{GPM18}} we
use martingale techniques and consider Poisson processes and associated
stopping times. The revenue ratio is then computed at once using Doob's
stopping time theorem. This last method has the advantage to fit the correct
profitability analysis, and to identify the source of the weakness of the
protocol. It allows a Bitcoin Improvement Proposal (BIP) to prevent the
attack. It also yields the mean duration time before the attack becomes
profitable. This last fact is out of reach with pure Markov chain models.

As usual, the relative hashrate of the honest miners (resp. attacker) is $p$
(resp. $q$) and $\gamma$ denotes its ``connectivity''.
We have $p + q = 1$, $q < \frac{1}{2}$ and $0 \leq \gamma
\leq 1$. We consider that whenever a competition occurs between two blocks or
two forks, there is a fraction $\gamma$ of the honest miners who mines on top
of a block validated by the attacker.

\section{Attack cycle and Dyck word}

An attack cycle for the SM strategy can be described as a sequence $X_0 \ldots
X_n$ with $X_i \in \{S, H\}$. The index $i$ indicates the $i$-th block
validated since the beginning of the cycle and letters $S, H$ determine the
miner who has discovered this block between the selfish miner ($S$) and the
honest miners ($H$).

\begin{example}
  The sequence $\text{SSSHSHH}$ means that the selfish miner has first
  validated three blocks in a row, then the honest miners have mined one, then
  the selfish miner has validated a new one and finally the honest miners have
  mined two blocks. At this point, the advantage of the selfish miner is only
  of one block. So according to the SM strategy, he decides to publish his
  whole fork and ends his attack cycle. In that case, we have $L = Z = 4$.
\end{example}

We are interested in the distribution of $L$.

\begin{theorem}
  We have $\mathbb{P} [L = 1] = p, \mathbb{P} [L = 2] = pq + pq^2$ and for
  $n \geqslant 3$, $\mathbb{P} [L = n] = pq^2  (pq)^{n - 2} C_{n - 2}$ where
  $C_n = \frac{(2 n) !}{n! (n + 1) !}$ is the $n$-th Catalan number.
\end{theorem}

\begin{proof}
  For $n \geqslant 3$, we note that $\{L = n\}$ is a collection of sequences
  of the form $w = \text{SSX}_1 \cdots X_{2 (n - 2)} H$ with $X_i \in \{ S, H
  \}$ for all $i$, such that if $S$ and $H$ are respectively replaced by the
  brackets ``(`` and ``)'' then, $X_1 \cdots X_{2 (n - 2)}$ is a Dyck word
  (i.e., balanced parentheses) with length $2 (n - 2)$ (see {\cite{K08}}). The
  number of letters ``$S$'' (resp. ``$H$'') in $w$ is $n$ (resp. $n - 1$). So,
  we get $\mathbb{P} [L = n] = p^{n - 1} q^n C_{n - 2}$ (see {\cite{K08}}).
  Finally, we note that $\{L = 1\} = \{H\}, \{L = 2\} = \{\text{SSH},
  \text{SHS}, \text{SHH}\}$. Hence we get the result.
\end{proof}

\begin{corollary}
  \label{el}We have $\mathbb{E} [L] = 1 + \frac{p^2 q}{p - q}$
\end{corollary}

\begin{proof}
  It comes from the well know relations
  \begin{align}
    \sum_{n \geqslant 0} p (pq)^n C_n & =  1  \label{pc}\\
    \sum_{n \geqslant 0} np (pq)^n C_n & =  \frac{q}{p - q}  \label{ec}
  \end{align}
  that have been already used and proved in {\cite{GPM18b}}.
\end{proof}

We can now compute the apparent hashrate.

\begin{theorem}
  \label{hashratesm}The long-term apparent hashrate of the selfish miner in
  Bitcoin is
  \[ \tilde{q}_B = \frac{[(p - q) (1 + pq) + pq] q - (p - q) p^2 q (1 -
     \gamma)}{pq^2 + p - q} \]
\end{theorem}

\begin{proof}
  If $L \geqslant 3$, then we are in the cases where all blocks validated by
  the selfish miner will end in the official blockchain. So, $Z = L$. If $L =
  1$, then $Z = 0$. Moreover, $Z (\text{SSH}) = Z (\text{SHS}) = 2$ and $Z
  (\text{SHH}) = 0$ (resp. 1) with probability $1 - \gamma$ (resp. $\gamma$).
  So,
  \begin{align*}
    \mathbb{E} [Z] & =  \mathbb{E} [L] - p - p^2 q \gamma - 2 p^2 q (1 -
    \gamma)\\
    & = \mathbb{E} [L] - (p + p^2 q + p^2 q (1 - \gamma))
  \end{align*}
  Using Corollary \ref{el} we get:
  \begin{align*}
    \frac{\mathbb{E} [Z]}{\mathbb{E} [L]} & =  \frac{p^2 q + p - q - (p -
    q)  (p + p^2 q + p^2 q (1 - \gamma))}{pq^2 + p - q}\\
    & =  \frac{[(p - q) (1 + pq) + pq] q - (p - q) p^2 q (1 - \gamma)}{pq^2
    + p - q}
  \end{align*}
  This is nothing but Proposition 4.9 from {\cite{GPM18}} which is itself
  another form of Formula (8) from {\cite{ES14}}.
\end{proof}

\section{Stubborn Mining}

We consider now two other block witholding strategies described in
{\cite{KMNS16}}. In the sequel, $C (x) = \frac{1 - \sqrt{1 - 4 x}}{2 x}$
denotes the generating series for the Catalan numbers $(C_n)_{n \geqslant 0}$.

\subsection{Equal Fork Stubborn Mining}

In this strategy, the attacker never tries to override the official
blockchain but when it is possible, he broadcasts the part of his secret fork
sharing the same height as the official blockchain as soon as the honest
miners publish a new block. The attack cycle ends when the attacker has been
caught up and overtaken by the honest miners by one block
{\cite{GPM18b,KMNS16}}. We show that the distribution of $L - 1$ is what we
have called a $(p, q)$-Catalan distribution of first type in {\cite{GPM18b}}.

\begin{theorem}
  \label{thepefsm}We have $\forall n \in \mathbb{N}, \mathbb{P} [L = n + 1] =
  p (pq)^n C_n$.
\end{theorem}

\begin{proof}
  Indeed, for $n \in \mathbb{N}$, $\{L = n + 1\}$ is a collection of sequences
  of the form $w = X_1 \cdots X_{2 n} H$ with $X_i \in \{ S, H \}$ for all
  $i$, such that if $S$ and $H$ are respectively replaced by the brackets
  ``(`` and ``)'' then, $X_1 \cdots X_{2 n}$ is a Dyck word with length $2 n$.
\end{proof}

\begin{corollary}
  \label{elefsm}We have $\mathbb{E} [L] = \frac{p}{p - q}$
\end{corollary}

\begin{proof}
  Obvious by (\ref{pc}) and (\ref{ec}).
\end{proof}

\begin{theorem}
  \label{thez}The long-term apparent hashrate of a miner following the
  Equal-Fork Stubborn Mining strategy is given by $\tilde{q} = \frac{q}{p} -
  \frac{(1 - \gamma)  (p - q)}{\gamma p}  (1 - pC ((1 - \gamma) pq))$.
\end{theorem}

\begin{proof}
  In an attack cycle, all the honest blocks except the last one have a
  probability $\gamma$ to be replaced by the attacker. So, $\mathbb{E} [Z|L =
  n + 1] = n + 1 - \frac{1 - (1 - \gamma)^{n + 1}}{\gamma}$. See Lemma B.1
  {\cite{GPM18b}}. Conditionning by $\{ L = n + 1 \}$ for $n \in \mathbb{N}$
  and using Theorem \ref{thepefsm}, we then get
  \[ \mathbb{E} [Z] = \frac{q}{p - q} - \frac{1 - \gamma}{\gamma}  (1 - pC ((1
     - \gamma) pq)) \]
  Hence we get the result.
\end{proof}

\subsection{Lead Stubborn Mining}

The strategy looks like the selfish mining strategy but here, the attacker
takes the risk of being caught up by the honest miners. When this happens,
there is a final competition between two forks sharing the same height. Once
the competition is resolved, a new attack cycles starts. In this case, the
distribution of $L - 1$ is what we have called a $(p, q)$-Catalan distribution
of second type {\cite{GPM18b}}.

\begin{theorem}
  \label{theplsm}We have $\mathbb{P} [L = 1] = p$ and for $n \geqslant 1$, $
  \mathbb{P} [L = n + 1] = (pq)^n C_{n - 1}$.
\end{theorem}

\begin{proof}
  Indeed, we have $\{ L = 1 \} = \{ H \}$ and for $n \in \mathbb{N}$, $\{L = n
  + 1\}$ is a collection of sequences of the form $w = \text{SX}_1 \cdots X_{2
  (n - 1)} \text{HY}$ with $X_1, \ldots, X_{2 (n - 1)}, Y \in \{ S, H \}$ and
  such that if $S$ and $H$ are respectively replaced by the brackets ``(`` and
  ``)'' then, $X_1 \cdots X_{2 (n - 1)}$ is a Dyck word with length $2 (n -
  1)$.
\end{proof}

\begin{corollary}
  \label{ellsm}We have $\mathbb{E} [L] = \frac{p - q + pq}{p - q}$
\end{corollary}

\begin{proof}
  Obvious by (\ref{pc}) and (\ref{ec}).
\end{proof}

By repeating the same argument as in the proof of Theorem \ref{thez} for the
computation of $\mathbb{E} [Z]$, we obtain the following theorem
{\cite{GPM18b}}.

\begin{theorem}
  The long-term apparent hashrate of a miner following the Lead Stubborn
  Mining strategy is given by $\tilde{q} = \frac{q (p + pq - q^2)}{p + pq - q}
  - \frac{pq (p - q)  (1 - \gamma)}{\gamma} \cdot \frac{1 - p (1 - \gamma) C
  ((1 - \gamma) pq)}{p + pq - q}$
\end{theorem}

We color the region $(q, \gamma) \in [0, 0.5] \times [0, 1]$ according to
which strategy is more profitable, and we obtain Figure 1 {\cite{GPM18b}} (HM
is the honest mining strategy).

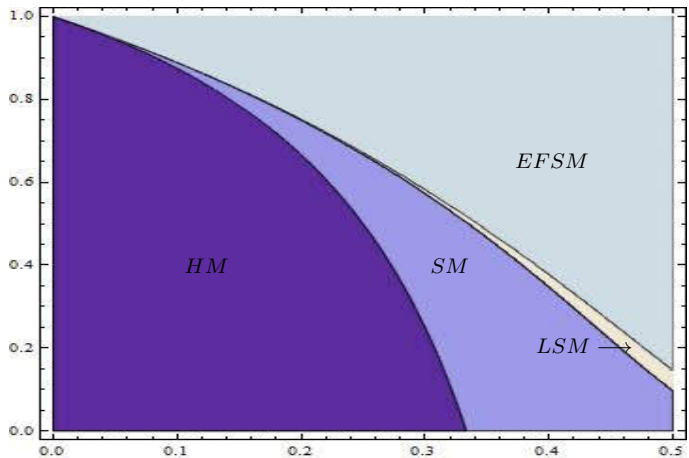
\begin{figure}[!ht]
   \includegraphics[height=6cm, width=9cm]{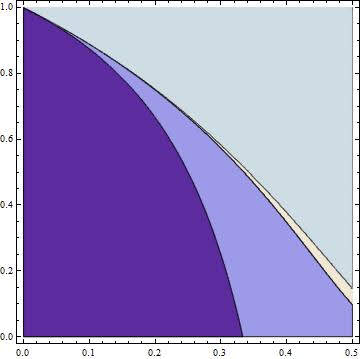}
   \put(-190, 70){\tiny{$HM$}}
   \put(-97,70){\tiny{$SM$}}
   \put(-57,40){\tiny{$LSM \longrightarrow$}}
   \put(-65,110){\tiny{$EFSM$}}
   \caption{
   Dominance regions in parameter space $(q,\gamma)$.
   \medskip}
\end{figure}


\begin{thebibliography}{1}
  \bibitem[1]{ES14}I.~Eyal, E.~Sirer. \textit{Majority is not
  enough: bitcoin mining is vulnerable.} International
  Conference on Financial Cryptography and Data Security,  pages  436--454,
  2014.
  
  \bibitem[2]{K08}T.~Koshy. \textit{Catalan Numbers with
  Applications}. Oxford University Press, 2008.
  
  \bibitem[3]{N08}S.~Nakamoto. \textit{Bitcoin: a peer-to-peer
  electronic cash system.} Bitcoin.org/bitcoin.pdf,
  2008.
  
  \bibitem[4]{GPM18}C.~Grunspan, R.~P{\'e}rez-Marco. \textit{On
  profitability of selfish mining.} ArXiv:1805.08281v2,
  2018.
  
  \bibitem[5]{GPM18b}C.~Grunspan, R.~P{\'e}rez-Marco. \textit{On
  profitability of stubborn mining.}ArXiv:1808.01041,
  2018.
  
  \bibitem[6]{GPM18e}C.~Grunspan, R.~P{\'e}rez-Marco. \textit{On
  profitability of trailing mining.} ArXiv:1811.09322,
  2018.
  
  \bibitem[7]{KMNS16}K.~Nayak, E.~Shi, S.~Kumar, A.~Miller.
  \textit{Stubborn mining: generalizing selfish mining and combining with
  an eclipse attack.} IEEE European Symp.  Security
  and Privacy,  pages  305--320, 2016.
\end{thebibliography}
\end{document}